\documentclass[lettersize,journal]{IEEEtran}

\def\BibTeX{{\rm B\kern-.05em{\sc i\kern-.025em b}\kern-.08em
    T\kern-.1667em\lower.7ex\hbox{E}\kern-.125emX}}
\usepackage{balance}

\IEEEoverridecommandlockouts

\def\BibTeX{{\rm B\kern-.05em{\sc i\kern-.025em b}\kern-.08em
    T\kern-.1667em\lower.7ex\hbox{E}\kern-.125emX}}
\usepackage{mathtools}
\usepackage[nodisplayskipstretch]{setspace}
\usepackage[utf8x]{inputenc}
\usepackage{tikz}

\DeclareMathOperator*{\argmin}{arg\,min}
\DeclareMathOperator*{\argmax}{arg\,max}

\def\blue{\textcolor{blue}}

\usepackage{subfigure}
\usepackage{moreverb}
\usepackage{epsfig}
\usepackage{amsmath,amssymb,amsthm,mathrsfs,amsfonts,dsfont}
\usepackage{tikz}
\usepackage{fancyhdr}
\usepackage{adjustbox,lipsum}
\usepackage[font={small}]{caption}
\usepackage{color,soul}
\usepackage{amsmath}
\usepackage{amsthm,amssymb,amsmath,bm}
\usepackage{subfigure}
\usepackage{amsfonts}
\usepackage{amsmath}
\usepackage{epsfig}
\usepackage{amssymb}
\usepackage{amsmath}
\usepackage{cite}
\usepackage{bbm}
\usepackage{cancel}
\usepackage{algorithmicx}
\usepackage{algpseudocode}
\usepackage[norelsize, linesnumbered, ruled, lined, boxed, commentsnumbered]{algorithm2e}
\usepackage[colorlinks,bookmarksopen,bookmarksnumbered,citecolor=blue,urlcolor=red]{hyperref}

\usepackage{nomencl}
\makenomenclature

\allowdisplaybreaks
\usepackage{setspace}
\usepackage{afterpage}

\newtheorem{Pro}{Proposition}

\begin{document}
\title{ 
%\huge 
Goal-Oriented Remote Tracking Through Correlated Observations in Pull-based Communications
}
\author{\IEEEauthorblockN{Abolfazl Zakeri, Mohammad Moltafet, and 
  \IEEEauthorblockN{Marian Codreanu}
  }
  %\vspace{-0.5 em}
 \thanks{
 A. Zakeri is with 
 CWC-RT,
   University of Oulu, Finland,
 Email: abolfazl.zakeri@oulu.fi. 
 M. Moltafet is with Department of Electrical and Computer Engineering University of California Santa Cruz, Email: mmoltafe@ucsc.edu.
  M. Codreanu is with Department of Science and Technology,
   Link\"{o}ping University, Sweden, 
   Email: marian.codreanu@liu.se.
  % This research has been financially supported by the Infotech Oulu, the Academy of Finland (grant 323698), and 6G Flagship program  (grant 346208). The work of M. Leinonen has also been financially supported in part by the Academy of Finland (grant 340171).
 }
}
\maketitle
	\begin{abstract}
We address the real-time remote tracking problem in a status update system comprising two sensors, two independent information sources, and a remote monitor. The status updating follows a pull-based communication, where the monitor commands/pulls the sensors for status updates, i.e., the actual state of the sources. \blue{We consider that the observations are \textit{correlated}, meaning that each sensor's sent data could also include the state of the other source due to, e.g., inter-sensor communications or overlapping monitoring regions.} The effectiveness of data communication is measured by a generic distortion, capturing the underlying application's goal. We provide optimal command/pulling policies for the monitor that minimize the average weighted sum distortion and transmission cost.
Since the monitor cannot fully observe the exact state of each source, we propose a partially observable Markov decision process (POMDP) and reformulate it as a belief MDP problem. We then effectively truncate the infinite belief space and transform it into a finite-state MDP problem, which is solved via relative value iteration. Simulation results show the effectiveness of the derived policy over age-based and deep-Q network baseline policies. 
%  \\
% \indent \textit{Index Terms--} Goal-oriented tracking, 
	\end{abstract}
\vspace{-1 em}
\section{Introduction} 
The increasing adoption of the Internet of Things~{(IoT)} and cyber-physical systems, such as smart factory/city/transportation, relies on real-time estimation and tracking of remotely monitored processes for tasks like data processing, actuation, planning, and decision-making. Pragmatic or goal-oriented communication, associated with Level C of communication problems \cite{Deniz_Semantic_JSAC}, is critical for these applications. It necessitates efficient communication system designs tailored to achieve certain end-user goals. One possible approach to measure the effectiveness of such systems is through distortion-based metrics that quantify the difference between the source information and its estimate at the remote monitor~\cite{Deniz_Semantic_JSAC,zakeri2024_phdthesis}.

Several studies in goal-oriented communications, particularly in real-time remote tracking and status updating, explored distortion-based performance measures, e.g.,~\cite{ELIF_Semantic_mag, Niko_mag_sem21, Niko_Remote_Recons, Aimin_goal_WC24, Zakeri_wcnc24, NiKo_Goal2, Peter_GoalOriented, zakeri2023semantic}. For instance,~\cite{ELIF_Semantic_mag, Niko_mag_sem21} highlighted the role of sampling and transmission policies in reducing reconstruction errors, and~\cite{Aimin_goal_WC24}  recently exploited similar ideas to introduce a unified performance metric for goal-oriented communications.
Furthermore, in \cite{Niko_Remote_Recons}, the performance of sampling and transmission policies for tracking two Markov sources was analyzed, and \cite{NiKo_Goal2} derived transmission policies under average resource constraints using constrained Markov decision processes (MDP) and drift-plus-penalty methods. Work~\cite{Zakeri_wcnc24} demonstrated that estimation strategies have a significant impact on performance. It revealed the limitations of approaches that assume a \textit{fixed} estimation strategy based on the last received sample, e.g.,~\cite{NiKo_Goal2}. Such fixed estimation strategies can degrade performance and lead to inaccurate assessments of sampling and transmission policies.

Despite notable progress in real-time remote tracking, the impact of \textit{correlation} in source observations remains underexplored. Correlated observations naturally arise in monitoring systems due to, e.g., spatial dependencies between sensors or inter-sensor communications. Although this concept, also referred to as correlated sources \cite{corl_source_Eytan}, has been studied in the context of the age of information (AoI), e.g., \cite{corl_source_Eytan,mattuk_corr}, its impact on goal-oriented communications--where the value and utility of source information are critical--was not thoroughly investigated. \blue{Furthermore, AoI-oriented works focus on age-related metrics that are agnostic to the actual value of the source information and estimation strategies, rendering them unsuitable for direct application in goal-oriented communications.} To address this gap, this letter focuses on real-time remote tracking considering \textit{correlation} in observations. %~(sensing). 

We consider a slotted system comprising two independent Markov sources and two sensors transmitting their observations to a remote monitor over error-prone communication channels, as illustrated in Fig.~\ref{Fig_SM_corr}. 
The correlation comes from the fact that an update from a sensor may also include the (exact) state of the other source \cite{Petar_obsr_ACM,corl_source_Eytan}. 
The system operates under a \textit{pull-based} communication model, e.g., \cite{agheli2024integrated, Ulukus_aoii_pull24}, where the monitor pulls/requests the state of the sources from the sensors. 
The proposed model applies to vehicular systems with inter-vehicle communications and industrial settings involving overlapping sensor observations.
The objective is to provide an optimal command strategy that minimizes the average weighted sum of distortion and transmission cost.  
Since the monitor does not have real-time knowledge about
%lacks direct, real-time knowledge of 
the sources, we propose a partially observable MDP (POMDP) to account for the uncertainty. 
We then formulate the belief-MDP problem, expressing the belief as a function of the AoI at the monitor. Noting that the belief values remain almost constant as the AoI increases, we upper bound the AoI and cast the problem as a finite-state MDP, which is solved using the relative value iteration algorithm (RVIA). 
The proposed optimal policy is evaluated against two widely used age-based benchmarks as well as a Deep Q-Network (DQN) policy. It consistently outperforms the age-based policies by a significant margin and also achieves better performance than the DQN policy across all scenarios. 
\\\indent 
The closest related work is~\cite{mattuk_corr}, which partly focuses on real-time error analysis for correlated observations in a queuing system. However, this letter differs fundamentally from \cite{mattuk_corr} in two key aspects:
1) The system model in \cite{mattuk_corr} is based on a queuing model, whereas this letter is a pull-based status-updating model, and
2) \cite{mattuk_corr} focuses on the \textit{analysis} of real-time error. In contrast, this letter focuses on the \textit{design} of pulling policies for a generic distortion metric, including real-time error,~among others.
\unskip
\vspace{-2 em}
\section{System Model and Problem Formulation }\label{sec_SM} \allowdisplaybreaks
\begin{figure}[t]
    \centering
    \hspace{-5 mm}
    \includegraphics[width= 0.44\textwidth]{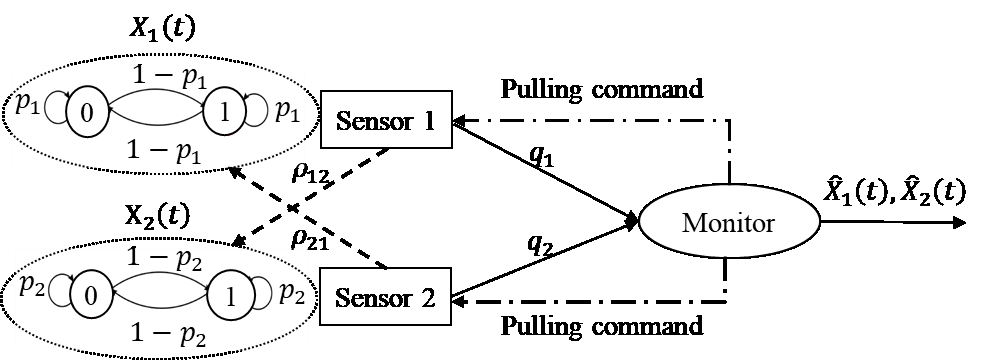}    
    \caption{System model.
    }
    \label{Fig_SM_corr}
    \vspace{-1.7 em}
\end{figure}
We consider a real-time tracking system consisting of two independent sources with corresponding sensors and a remote monitor, as shown in Fig.~\ref{Fig_SM_corr}. Time is divided into discrete time slots, i.e., ${ t\in\{0,1,\ldots\} }$.  
%\\\indent 
The sources are indexed by $i \in \{1,2\}$. Each source~$i$ is modeled as a binary symmetric discrete-time Markov chain~${X_i(t)=\{0,1\}}$ with the self-transition probability $p_i$. We choose a binary symmetric Markov model for the sake of presentation's clarity. Our proposed solution, however, can be easily extended to a generic finite-state Markov with an arbitrary number of states and non-symmetric transition probabilities.
\\\indent
Each source is observed by its dedicated sensor and potentially by the other sensor(s). This could happen due to direct communication between sensors (e.g., in vehicular communications) or overlapping sensor views (e.g., in smart factories where robots working in shared areas may detect each other's positions or other related information). At each time slot, each sensor~$i$ collects information about the state of source~$i$. Additionally, with probability $\rho_{ij}$, the measurement collected by sensor~$i$ also includes information about the current state of source~$j$,~$X_j(t)$.  
The overall correlation structure is captured by a matrix $\boldsymbol{\mathrm{P}}\triangleq [\rho_{ij}]$, where $\rho_{ij}\in[0,1]$ represents the probability that sensor~$i$ observes source~$j$, with $\rho_{ii} = 1$ for all~$i$. This type of correlation model is also used in~\cite{corl_source_Eytan, mattuk_corr}.  
\\\indent
The status updating follows a pull-based protocol. \blue{At each time slot~$t$, the monitor, which is the controller, can pull (command) a sensor.} The selected sensor then sends the status data to the monitor. Let $a(t)\in\{0,1,2\}$ denote the command action of the monitor at slot $t$, where $a(t)=0$ means the monitor is idle, and ${a(t)=i,~i=1,2,}$ means the monitor commands source~$i$.  
\\\indent
We assume an error-prone communication channel between each sensor and the monitor. Each transmission from sensor~$i$ to the monitor occupies one time slot and is successfully received with probability~${q_i }$, referred to as the reception success probability.  
\\\indent
The monitor generates a real-time estimate of each source~$i$. Let us denote the estimate of source~$i$ at slot $t$ by $\hat{X}_i(t)$ for all~$i$. Noting that the construction of the estimate $\hat{X}_i(t)$ depends on the design performance metric, next, we first define our performance metric and then our estimation strategy.
\\\indent
As a goal-oriented metric, we consider a generic distortion measure defined~as:\footnote{\blue{Note that this distortion is \textit{asymmetric and state-dependent}, capturing the \textit{significance} of any mismatch between the actual sources' state and their estimates,  based on the fact that in some applications, \textit{not
all errors have the same impact}. For instance, in smart grid outage detection, where power line status is Normal (1) or Faulty/Outage (0), missing a fault (0→1) risks blackouts or equipment damage, while falsely detecting a normal state (1→0) triggers unnecessary inspections. It is worth noting that the Age of Incorrect Information (AoII)~\cite{Tony_1} is another potentially relevant metric. However, AoII primarily focuses on the monotonically increasing \textit{time penalty} associated with system mismatches, rather than the asymmetric error penalty, which may be less critical for such scenarios.  
}}
\begin{equation}\label{eq_distor_fun}
    d_i(t) \triangleq f_i\left( X_i(t), \hat{X}_i(t) \right),
\end{equation}  
where the function ${ f_i: \mathcal{X}_i \times \mathcal{X}_i \rightarrow \mathbb{R} }$ could be any bounded function, i.e., ${|f(.)|<\infty}$ \cite{zakeri2023semantic}, where $\mathcal{X}_i\in\{0,1\}$ is the state space of each source~$i$. 
 The function $f_i(\cdot)$ can represent different estimation error metrics
 %. For example, it could be the squared estimation error, $d_i(t) = (X_i(t) - \hat{X}_i(t))^2$, or
 such as the real-time error, ${ d_i(t) = \mathds{1}_{\{X_i(t) \neq \hat{X}_i(t)\}} }$, where the binary indicator function $\mathds{1}_{\{.\}}$ equals one if the input argument holds true. 
\\ \indent
As shown in \cite{Zakeri_wcnc24}, the estimation strategy, i.e., the construction of $\hat{X}_i(t)$, has a significant impact on the scheduling policy and performance. We propose a \textit{minimum mean distortion} estimation, where distortion could be any function as defined in~\eqref{eq_distor_fun}. \blue{To compute the estimate at slot $t$, we require the probability distribution over the state space of the sources based on potentially all history of received samples and their time stamps until. 
Since the sources are Markovian, the most recent samples available at the monitor, denoted by $\bar{X}_i(t)$ for each source $i$, along with their respective ages, $\delta_i(t)$, provide sufficient information to obtain the estimate.} We assume that in case of a successful status update from a sensor during slot~$t-1$, the most recent sample at slot~$t$ will be updated to the source state at $t-1$, before the source state transition at~$t$.    
{\color{blue}The estimate can be formally expressed as 
\begin{align} 
    \hat{X}_i(t) &= \argmin_{Y_i(t)\in\mathcal{X}_i }~
    \mathbb{E}\!\left\{f_i\!\left(X_i(t),\,Y_i(t)\right)\right\}, \\
    &= \argmin_{Y_i(t)\in\mathcal{X}_i }~
    \sum_{x_i\in\mathcal{X}_i} 
    \Big(\Pr\{X_i(t) = x_i \mid \bar{X}_i(t),\delta_i(t)\} \notag \\
    &\hspace{4.5em}\times f(X_i(t)=x_i, Y_i(t))\Big).
\end{align}
}
% \begin{equation}
%     \hat{X}_i(t) = \argmin_{Y_i(t)\in\mathcal{X}_i }~ \Bbb{E}\left\{f_i\left(X_i(t),\,Y_i(t)\right)\right\},
% \end{equation}
% where the expectation is taken with respect to the probability distribution over~$X_i(t)$.
%\\\indent
At each time slot, the objective is to determine the optimal command action $a(t)$ at the monitor that minimizes a goal-oriented cost function. 
We define a cost function composed of two parts: a generic distortion term that reflects application-level performance objectives, and a transmission cost incurred by the radio resource usage for each status update. 
Our goal is to derive an optimal control policy for the monitor that minimizes the long-term average of a weighted sum of the distortion and transmission costs. This is formulated as a stochastic control problem below:
       \begin{align} \label{Org_P1}
          {\mbox{minimize}}~   &
           \limsup_{T\rightarrow \infty}\frac{1}{T}   \!\sum_{t=1}^T  \Bbb{E}\!\left\{ \sum_i w_id_i(t) \!+\! \alpha \mathds{1}_{\{a(t) \neq 0\} \!} \right\},
                   \end{align}
        	%	\end{subequations}
where the variables are ${ \{a(t)\}_{t=1,2,\ldots} }$. Here, $w_i$ denotes the weight assigned to source $i$ representing its (relative) importance, $\alpha$ is a coefficient that emphasizes the transmission cost, and $\mathbb{E}\{\cdot\}$ is the expectation operator taken over the randomness of the system (due to the sources, correlation, and wireless channel) and the (possibly randomized) action selection of $a(t)$.  
\vspace{-1 em}
\section{A Solution to Problem \eqref{Org_P1}}
This section proposes a solution to  problem~\eqref{Org_P1}. 
The sources $X_i(t)$ are not observable to the controller, and thus, we first model problem~\eqref{Org_P1} as
a POMDP. Then, we cast the POMDP into a finite-state MDP problem and solve it using RVIA.
\\\indent
The POMDP is described by the following elements:   
\\
$\bullet$ \textit{State:}
Let ${\delta_i(t)\in\{1,2,\dots\}}$ denote the age of the most recent sample $\Bar{X}_i(t)$ of source~$i$ \textit{available} at the monitor at slot~$t$.  
The state at slot $t$ is ${ s(t) = \left\{X_i(t), \bar{X}_i(t), \delta_i(t) \right\}_{i=1,2 } }$. 
The most recent samples $\bar{X}_i(t)$ and their  ages~$\delta_i(t)$ are included in the state because they are required to construct the estimate~$\hat{X}_i(t)$, hence the cost function. 
% ${ s(t) = \left(X_1(t), X_2(t), \bar{X}_1(t), \bar{X}_2(t), \right) }$.
\\
$\bullet$ \textit{Observation:} The observation (at the monitor) at slot $t$ is ${O(t) = \left\{ \bar{X}_i(t), \delta_i(t) \right\}_{i=1,2 }}$. 
\\ 
$\bullet$ \textit{Action:} The action at slot $t$ is $a(t)\in\mathcal{A}$, where $\mathcal{A}=\{0,1,2\}$ is the action space as defined in~Section~\ref{sec_SM}.
\\
$\bullet$ \textit{State transition probabilities:} The transition probabilities from  current state ${s= \left\{ X_i, \bar{X}_i, \delta_i \right\}_{i=1,2}}$ to next state
${s'= \left\{ X'_i, \bar{X}_i', \delta_i' \right\}_{i=1,2}}$
under a given action $a$ is denoted by  
$
{
        \Pr\{s'\,|\,s, a\} }
  $  
and it can be expressed as
\begin{equation}
    \begin{aligned}
        \Pr\{s'\,|\,s, a\}  &  \overset{(a)}{=}
 \Pr\left\{\bar X_1',\,\bar X_2',\,\delta_1',\,\delta_2'\,\big|\,s,\,a\right\}  
\Pr\{X_1',\,X_2'\,\big|\,s,\,a\}
\\& \overset{(b)}{=}
\Pr\left\{\bar X_1',\,\bar X_2',\,\delta_1',\,\delta_2'\,|\,s,\,a\right\} 
 \prod_i 
\Pr\{X_i'\,|\,X_i\}, %\nonumber
    \end{aligned}
    %\vspace{-1ex}
\end{equation}
where (a) follows from the fact that, given the current state $s$, the next most recent sample at the monitor, $\bar X_i'$, is conditionally independent of the next state of the source, $X_i'$, and (b) follows from the mutual independence between the sources.
Moreover, $\Pr\left\{\bar X_1',\,\bar X_2',\,\delta_1',\,\delta_2'\,|\,s,\,a\right\} $
is given by \eqref{eq_state_trans_xbardelta} at the top of the next page, and is computed by considering all possible outcomes for each action, incorporating the probabilistic events induced by channel reliabilities and source correlation probabilities. %, as explained in \cite[p. 3]{zakeri2025goal_full}. // for the ArXiv Version
 To explain, line 1 of \eqref{eq_state_trans_xbardelta} corresponds to the idle action, where no status update is transmitted. Lines 2–3 represent cases where either sensor 1 (line 2) or sensor 2 (line 3) attempts to transmit an update, but the transmission fails. Lines 4 and 6 correspond to successful updates that include the state of both sources. Lines 5 and 7 represent successful updates that only include the state of the transmitting sensor, without information from the other source. Finally, the last line of \eqref{eq_state_trans_xbardelta} accounts for all remaining transitions in the state space that are not explicitly covered in the previous lines. %// moved to the full version 
Finally, we have 
\begin{equation}\label{eq_xdyn}
     \begin{array}{cc}
        \Pr\{X_i'\,|\,X_i\} =\left\{ 
  \begin{array}{ll}
  p_i,  & \text{if} ~ X_i' = X_i
  \\
  1-p_i,   & \text{if} ~ X_i' \neq  X_i
     % \\
     % 0,  & \text{otherwise}.
    \end{array}
    \right.
     \end{array}.
 \end{equation}
 %---
 \begin{figure*}[t]
 \begin{equation}\label{eq_state_trans_xbardelta}
     \begin{array}{ll}
&        \Pr\left\{\bar X_1',\,\bar X_2',\,\delta_1',\,\delta_2'\,|\,s,\,a\right\}    
        =
        % \\&                                                                    
        \left\{ 
  \begin{array}{ll}
  1,  & \text{if} ~ a = 0,~\bar X_1' = \bar X_1,~\bar X_2' = \bar X_2,~\delta_1' = \delta_1+1,~\delta_2' = \delta_2+1,
  \\
  1-q_1,   & \text{if} ~ a = 1,~\bar X_1' = \bar X_1,~\bar X_2' = \bar X_2,~\delta_1' = \delta_1+1,~\delta_2' = \delta_2+1,
    \\
  1-q_2,   & \text{if} ~ a = 2,~\bar X_1' = \bar X_1,~\bar X_2' = \bar X_2,~\delta_1' = \delta_1+1,~\delta_2' = \delta_2+1,
  \\
  q_1 \rho_{12},   & \text{if} ~ a = 1,~\bar X_1' = X_1,~\bar X_2' = X_2,~\delta_1' = 1,~\delta_2' = 1,
  \\
  q_1 (1-\rho_{12}),   & \text{if} ~ a = 1,~\bar X_1' = X_1,~\bar X_2' = \bar X_2,~\delta_1' = 1,~\delta_2' = \delta_2+1,
    \\
  q_2 \rho_{21},   & \text{if} ~ a = 2,~\bar X_1' = X_1,~\bar X_2' = X_2,~\delta_1' = 1,~\delta_2' = 1,
  \\
  q_2 (1-\rho_{21}),   & \text{if} ~ a = 2,~\bar X_1' = \bar X_1,~\bar X_2' = X_2,~\delta_1' = \delta_1+1,~\delta_2' = 1,
     \\
     0,  & \text{otherwise}.
    \end{array}
    \right.
     \end{array}
 \end{equation}
 \hrule
 \vspace{-1 em}
 \end{figure*}
$\bullet$
\textit{Observation function:}
 The observation function is
 the probability distribution function of $o(t)$ given state $s(t)$ and action $a(t-1)$, i.e., ${ \Pr\{o(t)\,\big|\, s(t),a(t-1)\}}$. 
 Since the observation is always part of the state,
 the observation function is deterministic, i.e., ${ \Pr\{o(t)\,\big|\,s(t),a(t-1)\}  =  \mathds{1}_{\{o(t)=\left\{ \bar{X}_i(t), \delta_i(t) \right\}_{i=1,2 }\}} }$. 
 \\
 $\bullet$ \textit{Cost function:} The immediate cost function at slot $t$ is 
 \begin{equation}\textstyle
      C(s(t),a(t)) = \sum_i w_id_i(t) + \alpha \mathds{1}_{\{a(t) \neq 0\}}.
 \end{equation} Now, with the POMDP specified above, we follow the belief MDP approach \cite[Ch. 7]{POMDP_AI} to achieve an optimal decision-making for the POMDP. Accordingly, in the sequel,   we transform the POMDP into a belief MDP. 
\\ \indent
Let  $I(t)$ denote the complete information state at slot $t$ consisting of \cite[Ch. 7]{POMDP_AI}:
 i)     the initial probability distribution over the state space,
ii)  all past and current observations,  
$o(0),\dots, o(t)$,
and iii) all past actions, $a(0), \dots, a(t-1)$.
\\
The belief is a conditional probability distribution over the state space of the sources.
% The belief at slot $t$ is defined by 
% \begin{align}
%  \hspace{-.8 em}  b_{nm}(t)&=\Pr\{X_1(t) = n,\,X_2(t) = m\big|\,I(t)\}, n,m\in\{0,1\}
%    \\ & = \Pr\{X_1(t) = n\big|\,I(t)\}\Pr\{X_2(t) = m\big|\,I(t)\},
% \end{align}
% where the last line is due to the conditional independence, as the sources are independent of each other. Therefore, with an abuse of notation and 
Having binary independent sources, we define the belief at slot $t$ for each source $i$ as 
\begin{align}\label{eq_blf}
    b_i(t)\triangleq \Pr\{X_i(t) = 1\,\big|\,I(t)\},~\forall\,i\in\{1,2\}. 
\end{align}
We shall derive the belief update at $t+1$ from the current belief after execution of action $a(t)$ and receiving observation $o(t + 1)$. 
The proposition below gives the belief update function. 
\begin{Pro}
\label{prop_blf_evl}
For each source $i$, given the belief $b_i(t)$ and observation $\bar{X}_i(t+1), \delta_i(t+1)$ the belief update is:
      \begin{equation}\label{eq_blf_evl}
  \hspace{-1.5 em}
     \begin{array}{ll}
 & b_i(t+1) =
 \\&
 \left\{ 
  \begin{array}{ll}
  p_i  & \text{if}~ \delta_i(t+1) = 1,~\bar{X}_i(t+1)=1 \\
   1-p_i  & \text{if}~ \delta_i(t+1) = 1,~\bar{X}_i(t+1)=0 , \\
    b_i(t)p_i + (1-b_i(t))\bar{p_i}, & \text{if}~ \delta_i(t+1) \neq 1.
\end{array}\right. 
    \end{array}
 \end{equation}
 where $\bar{p_i} \triangleq 1-p_i$.
\end{Pro}
\begin{proof} It follows directly from the belief definition in~\eqref{eq_blf} for $t+1$, and its details derivations are in~Appendix~\ref{App_blfUpdate}. 
%\cite[Appendix A]{zakeri2025goal_full}.
 To explain when $\delta_i(t+1) = 1$, the monitor has exact knowledge of the source’s previous state, i.e., slot $t$, and the remaining uncertainty is limited to a single transition in the Markov chain. When there is no update, i.e., $\delta_i(t+1) \neq 1$, the belief evolves over the current belief based on the transition possibilities in the Markov chain. % needs more space, available in the full version 
\end{proof}
The belief state is then defined by
\begin{equation}\label{eq_blf_state}
    z(t) \triangleq \left\{b_i(t),\bar{X}_i(t),\delta_i(t) \right\}_{i=1,2}
\end{equation}
\blue{Having the belief state defined, we can formulate the belief MDP problem.} However,  solving the resulting (belief) MDP problem is challenging because the belief is a continuous state variable, and thus, the state space of the problem is infinite. Nonetheless, from the evolution of belief in \eqref{eq_blf_evl} we can observe that the last line in the equation is the propagation of uncertainty in the associated Markov chain, where the value of $\delta_i(t)$ indicates how many transitions happened from the last time the source state was $\bar{X}_i(t)$. This argument suggests that we can equivalently rewrite the belief in~\eqref{eq_blf_evl}, based on the $N$-step transition probabilities formula \cite{Kam_Towards_eff_2018}, as follows:
\begin{equation}\label{eq_blf_Nstep}
    \hspace{-1 em}
     \begin{array}{ll}
 & b_i(t+1) =
% \\&
 \left\{ 
  \begin{array}{ll}
  0.5\left( 1+(2p_i-1)^{\delta_i(t+1)}\right)   & \text{if}~\bar{X}_i(t+1)=1 \\
0.5\left( 1-(2p_i-1)^{\delta_i(t+1)}\right)   & \text{if}~\bar{X}_i(t+1)=0.
\end{array}\right. 
    \end{array}
 \end{equation}
From \eqref{eq_blf_Nstep}, we observe that for sufficiently large values of $\delta_i(t)$, i.e., $\delta_i(t) \ge N$, the belief approaches $0.5$ as $\delta_i(t)$ increases.
This implies that we can re-define the state~\eqref{eq_blf_state}~as 
\begin{align}\label{eq_truncblf_state}
   \underline{z}(t) &=
  \\& \nonumber
  \left\{\bar{X}_i(t),\delta_i(t)\right\}_{i=1,2},\,\bar{X}_i(t)\in\{0,1\},\delta_i(t)\in\{1,\dots,N\},
\end{align}
 and formulate a finite-state MDP problem.
\\\indent 
Note that the action of the above-mentioned finite-state MDP problem remains the same as that of the POMDP. Moreover, the state transition probabilities are given by
\begin{equation}
    \begin{aligned}
        \Pr\{ \underline{z}'\,|\,\underline{z}, a \} 
        &= \sum_{X_1,X_2\in\{0,1\}}
        \Pr\{\underline{z}'\,|\, \underline{z},a,X_1,X_2\} \\
        &\quad \times \Pr\{X_1,X_2\,|\, \underline{z},a\}.
    \end{aligned}
\end{equation}
where $\Pr\{\underline{z}'\,|\, \underline{z},a,X_1,X_2\}$ can be obtained by \eqref{eq_xdyn}, and due to the independency between the dynamic of sources, $ \Pr\{X_1, X_2\,|\, \underline{z} ,a\}$ can be written as $\textstyle {\Pr\{X_1,X_2\,|\, \underline{z},a\}=\prod_{i\in\{1,2\}} \Pr\{X_i\,|\, \underline{z},a\}  }$, where
\begin{equation}
   \hspace{-1 cm}  \begin{array}{ll}
& \Pr\{X_i\,|\,\underline{z},a\}  = 
\\&
 \left\{ 
  \begin{array}{ll}
  0.5\left(1+(2p_i-1)^{\delta_i(t)} \right), & \text{if}~~ X_i = 1,~\bar{X}_i=1,
  \\
   0.5\left(1-(2p_i-1)^{\delta_i(t)} \right), & \text{if}~~ X_i = 0,~\bar{X}_i=1,
  \\
  0.5\left(1+(2p_i-1)^{\delta_i(t)} \right), & \text{if}~~ X_i = 0,~\bar{X}_i=0,
  \\
 0.5\left(1-(2p_i-1)^{\delta_i(t)} \right), & \text{if}~~ X_i = 1,~\bar{X}_i=0.
\end{array}\right. 
    \end{array}
 \end{equation}
Finally, the immediate cost function is 
\begin{align}
    C(\underline z(t)) &= \textstyle\sum_i w_i 
    \Big[ b_i(t) f(1,\hat{X}_i(t))  
    + (1 - b_i(t)) f(0,\hat{X}_i(t)) \Big]  \nonumber \\
    &\quad + \alpha \mathds{1}_{\{a(t) \neq 0\}}.
\end{align}
where $b_i(t)$ is given by \eqref{eq_blf_evl}, and the estimate $\hat{X}_i(t)$ is obtained according to minimum mean distortion estimation, using the ages $\delta_i(t)$ and the last sample $\bar{X}_i(t)$ given in the state $\underline{z}(t)$.
\\\indent
Due to the source dynamics and the chance of loss in the channels, it is not difficult to see that the above-formulated MDP is unichain. Furthermore, the MDP is also aperiodic. Shortly, these follow because the state $(0,0,N,N)$ can be accessible from any state, including itself, under any stationary deterministic policy.  By unichain and aperiodicity properties of the MDP, we apply RVIA to obtain an optimal policy, which is guaranteed to converge~\cite{Zakeri_Journal_Relay}. 

{\color{blue}
The RVI algorithm turns the Bellman optimality equation into the following iterative procedure that for all state $\underline{z}$ in the state space $\underline{\mathcal{Z}}$ and for iteration index ${ n=1,2,\dots }$, we have
\begin{equation}\label{Eq_RVIal}
\begin{array}{ll}
&
\displaystyle
V^{n+1}(\underline{z}) =
\min_{a \in \mathcal{A}}
\left\{C(\underline{z})+ \underset{\underline{z}^{\prime} \in \underline{\mathcal{Z}}}{\sum} 
%\sum_{s^{\prime} \in \mathcal{S}} 
\operatorname{Pr}\left\{\underline{z}^{\prime} \mid \underline{z}, a\right\} 
h^n(\underline z')\right\}
\\
&
\displaystyle
h^n(\underline z) = V^n(\underline{z})
- V^n(\underline{z}_{\mathrm{ref}}),
%\forall \underline{z}\in\underline{\mathcal{Z}}, 
\end{array}
\end{equation}
%}
% // points taken from Berskas, dyanmic programming 
\begin{comment}
    1) Prop. 4.2.3: WA condition holds then the optimal averga cost is independ of the initial state.
    2) P. 198, Under WA condition (or whatever that allows the independency of the optimal values of the initial state), the optimality equations reduces to a single equation, Bellman's equation, Eq. 4.47  
    3) Prop 4.1.8, prove the optimality equations 
\end{comment}
where $\underline{z}_{\mathrm{ref}} \in \underline{\mathcal{Z}}$ is an arbitrarily chosen reference state, and we initialize $ { V(\underline{z}) = 0},~\forall\, \underline{z}\in \underline{\mathcal{Z}}, $ and repeat the process until it converges. 
Once the iterative process of the RVI algorithm converges, the algorithm returns an optimal policy 
$
%\pi^{*} =
\{a^{*}(\underline{z})\}$ 
and the optimal value, which equals $V({\underline{z}_{\mathrm{ref}}})$. A full description of the RVI algorithm is given in Algorithm~\ref{alg_RVI}, where $\epsilon$ is a small positive number.

\begin{algorithm}
    \SetKwInOut{Inputi}{Initialize}
    %\KwIn{Input}{Input}
    \SetKwInOut{run}{RUN}
     \SetKwInOut{output}{Output}
     \SetKwInOut{Output}{Output}
     %\SetKwInOut{AuxV}{Auxiliary Variables}
     \SetKwComment{Comment}{/*}{ }
     \SetKwRepeat{Do}{do}{while}
    \Inputi{ $\underline{z}_{\mathrm{ref}}$, $\epsilon$, $n=0$, set $h^0(\underline{z}) =0$ for all $\underline{z}\in\underline{\mathcal{Z}}$}
    %\Comment{RVI stopping criterion}
    \Do{$\displaystyle\max_{\underline{z}\in\mathcal{\underline{Z}}}|h^{n}(\underline{z})-h^{n-1}(\underline{z})|\geq\epsilon$}{
     $n = n+1$\\
    \For{$
    \displaystyle
    \underline{z}\in\mathcal{\underline{Z}}$}{
    $\displaystyle V^{n}(\underline{z}) = \min_{a\in \mathcal{A}}\left[C(\underline{z}) +\sum_{\underline{z}'\in \mathcal{\underline{Z}}}{\Pr}(\underline{z}' \mid \underline{z},a)h^{n-1}(\underline{z}')\right]$\\
    $h^{n}(\underline{z}) = V^{n}(\underline{z})-V^{n}(\underline{z}_{\mathrm{ref}})$\\
    %$ax2(a)=\infty$ for all $a\in\{0,1\}$
    } }
%\doWhile{$\displaystyle\max_{l\in\mathcal{L}}|h^{i}(l)-h^{i-1}(l)|\geq\epsilon$}
    \Comment{Generate a (deterministic)  policy}
    
    $\displaystyle \pi(\underline{z})=\argmin_{a\in \mathcal{A}}\left[C(\underline{z}) +\sum_{\underline{z}'\in \mathcal{\underline{Z}}}{\Pr}(\underline{z}' \mid \underline{z},a)h(\underline{z}')\right],~{\forall \underline{z}\in \mathcal{\underline{Z}}}$
    %$\pi^* = \pi$
    % %\Comment{An optimal policy for given $\bar \beta$}
    % \For{$l\in\mathcal{L}$}{$\pi^*(l)=\argmin_{a\in \mathcal{A}_l}\big[C(l) +\sum_{l'\in \mathcal{L}}\mathrm{Pr}(l' \mid l,a)h^i(l')\big]$}
 
    \KwOut{An optimal policy $ \pi^*=\pi$, the optimal value $C^* = V(\underline{z}_{\mathrm{ref}})$} 
    \caption{\blue{The RVIA algorithm} }
    \label{alg_RVI}
\end{algorithm}
}

%The complexity analysis of the derived policy can be found in \cite[p. 4]{zakeri2025goal_full}.
% \\\indent
\textit{Complexity Analysis:} The computational complexity of the derived command policy involves two phases: (i) an offline phase, in which RVIA is executed to obtain the policy in the form of a lookup table, and (ii) an online phase, where the obtained policy is applied in real time. The offline complexity is dominated by the computational cost of the RVIA, which is on the order of $|\mathcal{A}||\mathcal{\underline{Z}}|^2$ \cite{Zakeri_Journal_Relay}, where $|\mathcal{A}| = 3$ denotes the size of the action space, and $|\mathcal{\underline{Z}}|=N^2 {M}^{2}$ is the size of the truncated state space, with $M$ being the number of the states of each source Markov chain.
 The complexity of the online phase is $\mathcal{O}(1)$, as it involves only a simple lookup operation in the precomputed table.

%\vspace{-2 em}
\section{Numerical Results}\label{sec_numres}
This section presents simulation results to demonstrate the effectiveness of the derived policy and the impact of various parameters on the system's performance.\footnote{The implementation code is available at https://github.com/AZakeri94/goal-oriented-correlated-observation.}
The RVIA stopping criterion is set to $10^{-3}$, the weights of both sources, $w_1$ and $w_2$, are set to $1$, and the distortion metric defaults to real-time error unless stated otherwise. We set $N=30.$\footnote{To ensure reliable results, $N$ should be chosen sufficiently large such that in general
$P_i^N$
approximates the steady-state distribution of the Markov chain associated with source $i$, where $P_i$ denotes its transition probability matrix.}
Other parameters are given in the caption of each figure.
\\\indent \blue{For benchmarking, we consider the following policies:  
1) \textit{Max-age-first}: This policy commands the sensor whose direct observing source has the highest (weighted) age at the monitor, that is, 
${i^*\triangleq \underset{i\in\{1,2\}}{\argmax}~w_i\delta_i}$},
%whose direct observing source has the highest age at the monitor,}  
2) \textit{Age-optimal}: This policy is derived by replacing the distortion, $d_i(t)$, with the AoI at the monitor, $\delta_i(t)$, in the objective function \eqref{Org_P1} and solving the resulting problem,
and 3) \textit{Deep Q-Network (DQN)}: We adopt a DQN policy for the belief MDP problem with state in~\eqref{eq_truncblf_state}. We consider a fully connected linear network with one hidden layer of size $256$, ReLU activation, and the Adam optimizer. The learning rate is set to $0.001$, the discount factor to $0.99999$, the batchsize is $64$,  and the model is trained for $200$ epochs with $300$ iterations per epoch. We then save the trained model, run it for the same duration, and record the average cost function over that period.
\\\indent 
We first analyze the effect of the coefficient factor $\alpha$ in~\eqref{Org_P1} (i.e., the transmission cost) on the average weighted total cost, as shown in Fig.~\ref{fig_alpha}. The results indicate that the proposed policy significantly outperforms the age-based baseline policies, particularly as the transmission cost increases. This improvement occurs because the max-age-first policy ignores transmission costs and the age-optimal policy is source-agnostic and may initiate transmissions even when the source state matches its estimate at the monitor. Additionally, Fig.~\ref{fig_alpha} shows that the performance of the distortion-based policies (i.e., RVIA and DQN) remains nearly constant beyond a certain point. This is because when the transmission cost is high, the optimal action for the monitor is to remain idle, leading to maximum distortion irrespective of~$\alpha$. 
%The distortion depends primarily on the source dynamics, $p_1$ and $p_2$. 
\begin{figure}[t!]
    \centering
    \includegraphics[width= 0.4\textwidth]{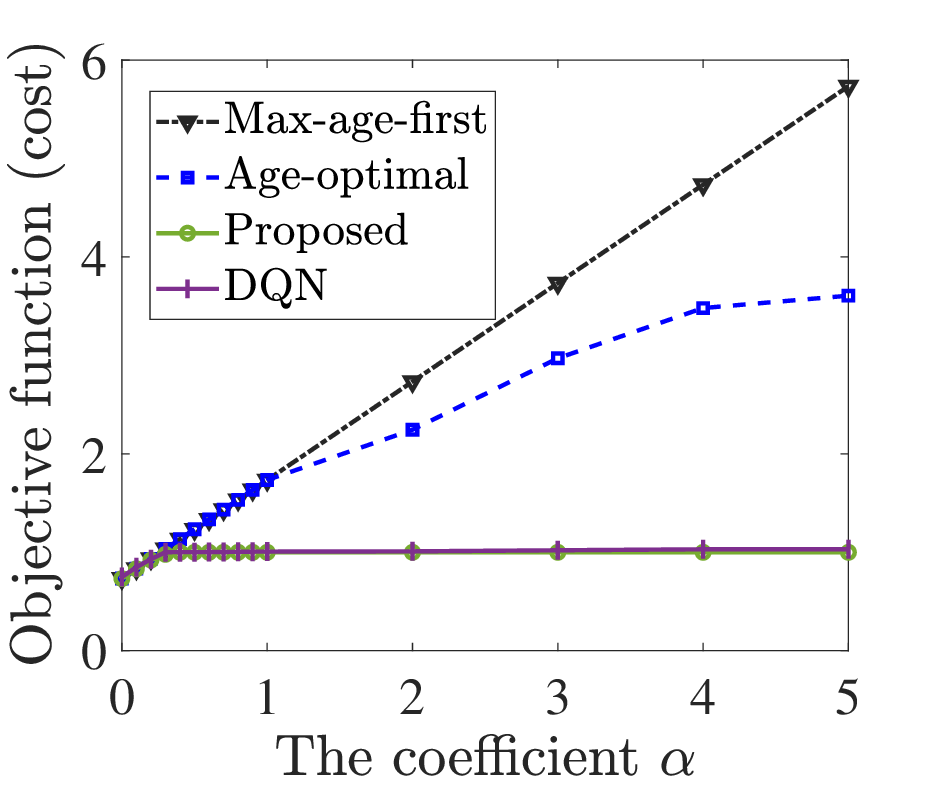}  % 0.25 for the jounral and 50 for ArXIv
    \vspace{-2 mm}
    \caption{The average cost (the objective function in \eqref{Org_P1}) vs. the coefficient $\alpha$ (transmission cost) for $p_1=0.7$, $p_2=0.7,$ $q_1=0.8$,  $q_2=0.6,$ $\rho_{12}=0.4$, and $\rho_{21}=0.7$.
    }
    \label{fig_alpha}
    \vspace{-6 mm}
\end{figure}
%===

Fig.~\ref{Fig_correlation} shows the average cost as a function of the observation correlations, assumed equal for both sources ($\rho_{12} = \rho_{21}$), for two different distortion functions:  1) the real-time error in Fig.~\ref{fig_crl_rte}, and 2) a distortion given by \eqref{eq_dis_ex} in Fig.~\ref{fig_crl_gnrdis}.  
The results indicate that as correlation probabilities increase, the average cost decreases for both distortion functions, particularly for the real-time error one. This is because a higher correlation provides a higher chance to update the monitor about the other source's state as well, which in turn corrects the sources' estimates and reduces the distortion.  
Furthermore, the figures show that for both real-time error distortion and the distortion in~\eqref{eq_dis_ex}, the age-based policies perform almost identically, while there is a notable gap between these baseline policies and the proposed optimal policy.  
\begin{equation}\label{eq_dis_ex}
 d_1 =
      \begin{bmatrix}
   0 & 3   \\ 1& 0& 
 \end{bmatrix},~~~~~~~~
\\
 d_2 =
      \begin{bmatrix}
0 & 1   \\ 5 & 0& 
 \end{bmatrix}.
 \end{equation}
\begin{figure}[t!]
%\centering
%\hspace{-0.3cm}
\subfigure[The real-time error] 
{
\includegraphics[width=0.4\textwidth]{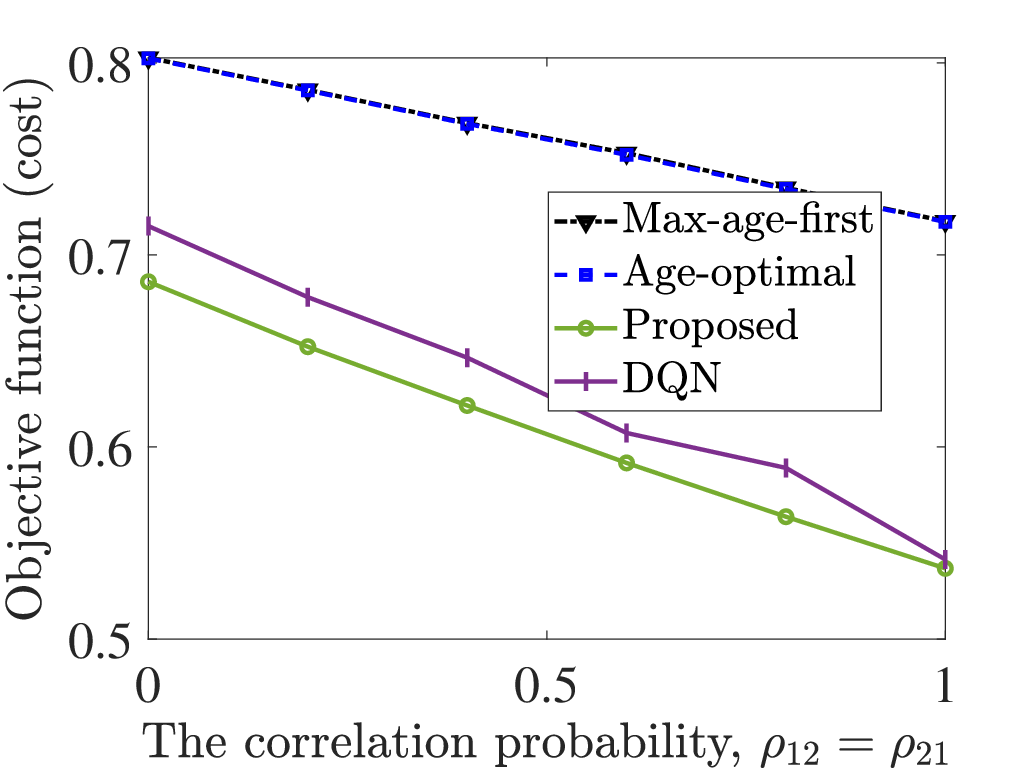} % 0.24 for the jounrnal
\label{fig_crl_rte}
}
%\vspace{-10 pt}
\hspace{-0.5cm}\subfigure[The distortion given by \eqref{eq_dis_ex}]{
\includegraphics[width=0.4\textwidth]{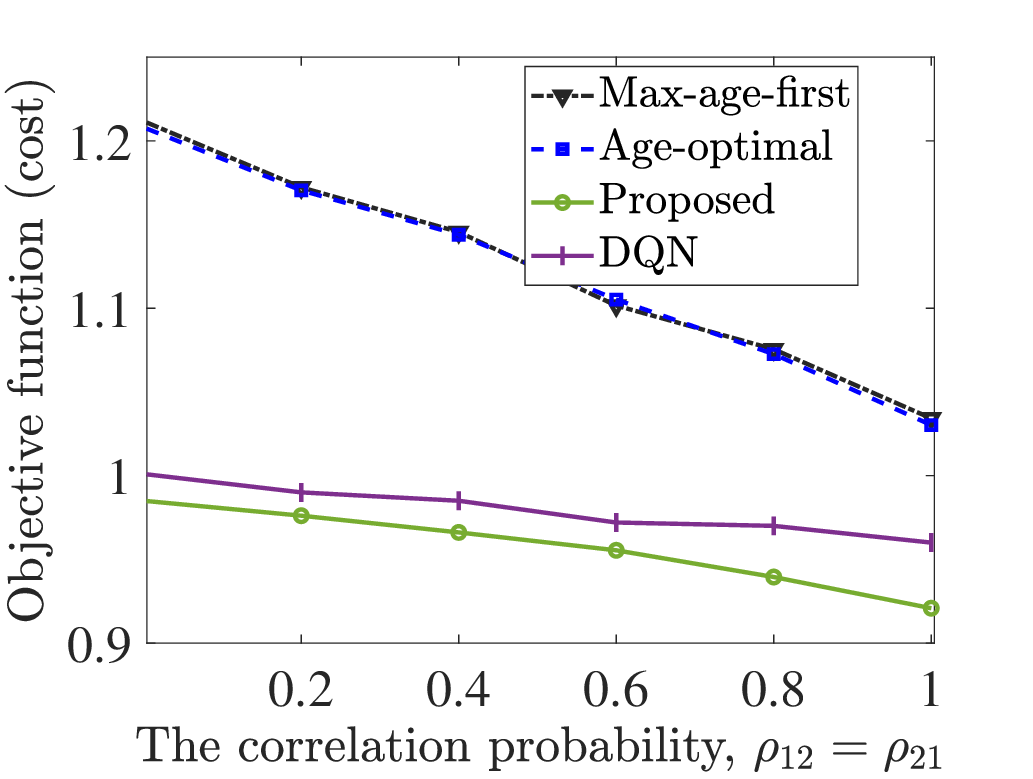}
\label{fig_crl_gnrdis}
}
\vspace{-2 mm}
\caption{The average cost vs. the correlation probabilities for different distortions for~$p_1=0.9$,~$p_2=0.9$,~$q_1=0.9$,~$q_2=0.9$,~and~$\alpha = 0.5$.
}
\label{Fig_correlation}
\vspace{-6 mm}
\end{figure}
\indent
Fig.~\ref{fig_selftran} illustrates the average cost as a function of the self-transition probabilities of the sources, $p_1$ and $p_2$, for different policies. The figure highlights a symmetric behavior of the cost function, with the maximum cost occurring at $p_1 = p_2= 0.5$, where the sources have the maximum entropy, making them harder to track accurately.  
This symmetric behavior is expected due to the trackability of sources that are either slow-varying, i.e., high $p_1$ and $p_2$, or fast-varying, i.e., low $p_1$ and $p_2$, \textit{provided that} an appropriate optimal estimation strategy is used. However, if the estimation strategy does not adapt to the source dynamics, such as using a fixed last-sample estimate,
the performance degrades significantly, particularly for the fast-varying sources.  
The performance gap between the proposed policy and the age-based baseline policies is most pronounced at $p_1= 0.5$ and $ p_2 = 0.5$, as the baseline policies are source-agnostic. 
In contrast, the DQN policy coincides with the optimal policy.
Nonetheless, the symmetric behavior of the age-based policies persists because we still consider minimum-distortion estimation at the monitor; otherwise, their performance would remain unchanged.
\begin{figure}[t!]
    \centering
    \includegraphics[width= 0.4\textwidth]{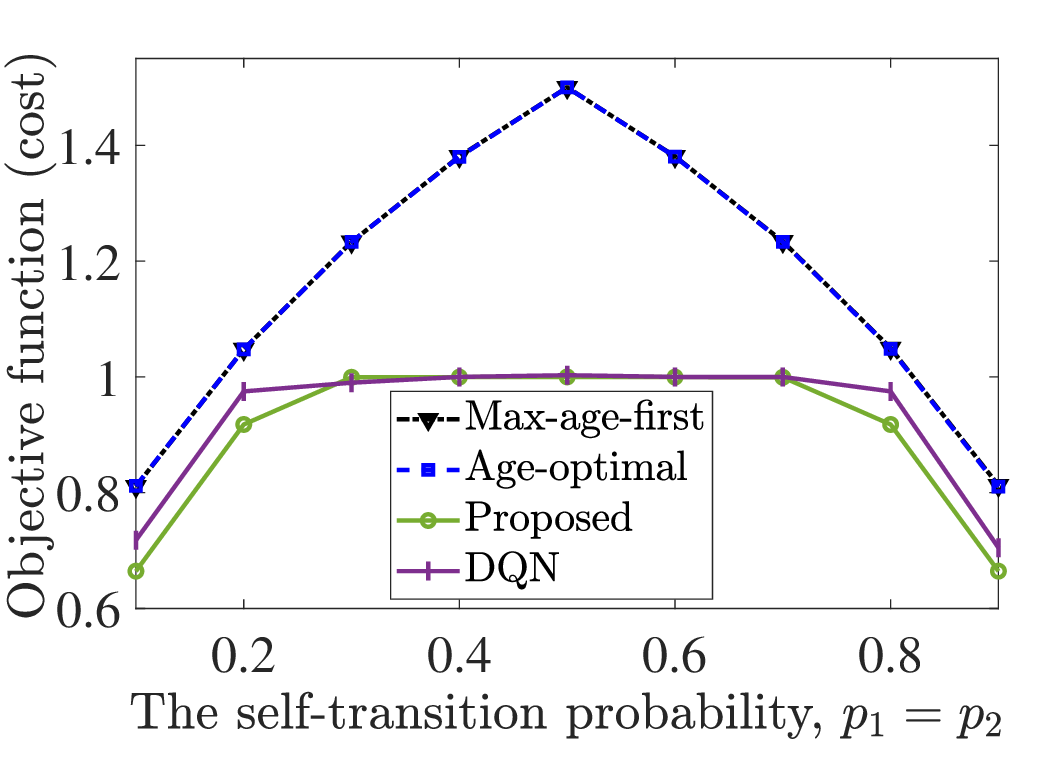}     % 0.245 for the journal
    \vspace{- 2 mm}
    \caption{The average cost vs. the self-transition probabilities of the sources for $q_1=0.8$,  $q_2=0.6,$ $\rho_{12}=0.4$, $\rho_{21}=0.7$, and $\alpha = 0.5$.
    }
    \label{fig_selftran}
    \vspace{-4 mm }
\end{figure}
\\\indent
%===
Fig.~\ref{fig_chnlrlb} examines the impact of channel reliability on the performance of different policies with $q_1 = q_2$. The results show a direct relationship between cost reduction and channel reliability: higher channel reliability leads to a lower cost. This is because, by increasing the reliability, the likelihood of successful transmissions increases, enabling the monitor to stay informed about the source states and track them more accurately.  
The figure also shows that the trend with varying channel reliabilities is similar across all policies, with an increasing performance gap between the age-based policies and both the proposed and DQN policies. 
\begin{figure}[t!]
    \centering
    \includegraphics[width= 0.4\textwidth]{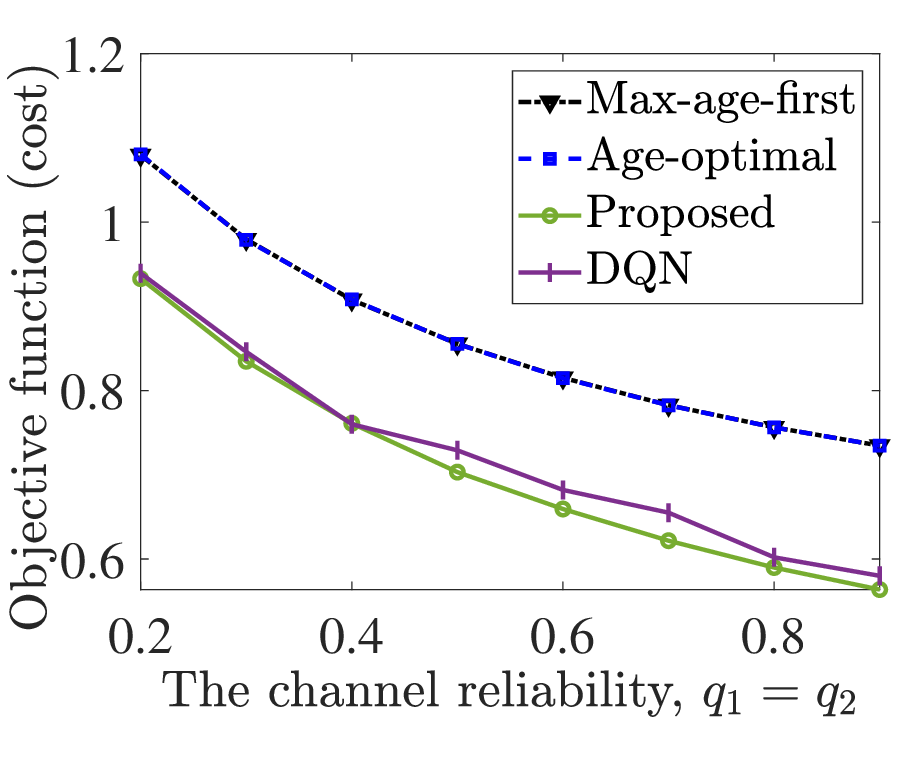}   %0.25 for journal
   \vspace{-2 mm}
    \caption{The average cost vs. the channel reliabilities for $p_1=0.9$,  $p_2=0.9$, $\rho_{12}=0.8$, $\rho_{21}=0.8$, and $\alpha = 0.5$.
    }
    \label{fig_chnlrlb}
     \vspace{- 7 mm}
\end{figure}
\section{Conclusions }\label{sec_conl}
We studied a goal-oriented real-time remote tracking problem in a status updating system with \textit{correlated observations}, comprising two independent Markov sources, two sensors, and a monitor.
We aimed to find the optimal pull policy for the monitor to minimize the weighted sum of distortion and transmission costs. Using a POMDP-based approach, we formulated the problem as a belief-MDP. Then, by expressing the belief as a function of AoI, we cast a finite-state MDP problem and solved it via RVI.
\\\indent
The simulation results showed that the derived policy outperforms the baseline policies, max-age-first and age-optimal, in cost reduction, while the baseline policies showed similar performance in most cases. Results also indicated that the correlation reduces distortion, which is symmetric around the self-transition probability being $0.5$. While we addressed correlation in the observations, exploring the correlation in the sources' dynamics remains an open problem for future work.
\vspace{-1 em}

\appendix %[Proof of Proposition \ref{Prop_ComMDP}]
\subsection{\blue{ Proof of Proposition \ref{prop_blf_evl} }}\label{App_blfUpdate}
We first begin with the belief definition at $t+1$ as:
\begin{align}
\nonumber
& b_i(t+1) = 
\\ &\Pr\left\{ X_i(t+1) =1 \mid I(t) = I, o(t+1) = o, a(t) = a \right\}
\\\label{eq_blf_prob} &
= \frac{ \Pr\left\{ X_i(t+1)=1, o(t+1)=o \mid I(t)=I, a(t)=a \right\} }{ \Pr\left\{ o(t+1)=o \mid I(t)=I, a(t)=a \right\} }
\end{align}
%  % \\&=
%  % \frac{ \Pr\left\{ X_i(t+1)=1, I(t), o(t+1), a(t) \right\} }{ \Pr\left\{ I(t), o(t+1), a(t) \right\} } 
% \\\label{eq_blf_prob}
% &= \frac{ \Pr\left\{ X_i(t+1)=1, o(t+1) \mid I(t), a(t) \right\} 
% %\cancel{\Pr\left\{ I(t), a(t) \right\}} 
% }
% { \Pr\left\{ o(t+1) \mid I(t), a(t) \right\} 
% %\cancel{ \Pr\left\{ I(t), a(t) \right\}}
% },
which follows directly from the definition of conditional probability. 
Notice that for notational simplicity, we omit explicit mention of the fixed values $I, o, a$, and write $\Pr\{X, o \mid I, a\}$ instead of $\Pr\{X_i(t+1)=1, o(t+1)=o \mid I(t)=I, a(t)=a\}$.
Next we calculate the numerator of \eqref{eq_blf_prob} as follows:
%\begin{figure*}
    \begin{align}
& \Pr\left\{ X_i(t+1)=1,\, o \mid I,\, a \right\} 
%\\ 
%&= \sum_{x\in\{0,1\}} \Pr\left\{ X_i(t+1)=1,\, o,\, X_i(t) = x \mid I,\, a \right\}
\\ \nonumber
&= \sum_{X_i\in\{0,1\}} \Pr\left\{ X_i(t+1) = 1 \mid X_i(t) = X_i,\, I,\, a,\, o \right\} 
\\
&\quad \times \Pr\left\{ X_i(t) = X_i \mid I,\, o,\, a \right\} \Pr\left\{ o \mid I,\, a \right\}
    \end{align}
%\end{figure*}
By plugging the above equation in \eqref{eq_blf_prob} and canceling out the terms $\Pr\left\{ o(t+1) \mid I(t), a(t) \right\}$, we obtain:
% \begin{align}
%    & b_i(t+1) =
%     \\ & \nonumber \sum_{x\in\{0,1\}}\Pr\left\{ X_i(t+1) = 1\mid X_i(t) = x, I(t), a(t), o(t+1) \right\} 
%     \\&
%      \times \Pr\left\{ X_i(t) =x \mid  I(t), o(t+1), a(t) \right\}
%      \\ & \nonumber = \sum_{x\in\{0,1\}} \Pr\left\{ X_i(t+1) = 1\mid X_i(t) = x \right\} 
%  \\& \label{eq_blf_finalsum}   \times \Pr\left\{ X_i(t) =x \mid  I(t), o(t+1), a(t) \right\},
% \end{align}
\begin{align}
   & b_i(t+1) =
    \\ & \nonumber \sum_{X_i\in\{0,1\}}\Pr\left\{ X_i(t+1) = 1\mid X_i(t) = X_i, I, a, o \right\} 
    \\&
     \times \Pr\left\{ X_i(t) = X_i \mid  I, o, a \right\}
     \\ & \nonumber = \sum_{X_i\in\{0,1\}} \Pr\left\{ X_i(t+1) = 1\mid X_i(t) = X_i \right\} 
 \\& \label{eq_blf_finalsum}   \times \Pr\left\{ X_i(t) = X_i \mid  I, o, a \right\},
\end{align}
where the first term is simply the source's transition probabilities given by
\begin{align}\label{eq_blf_fist}
\begin{array}{ll}
 &  \Pr\left\{ X_i(t+1) = 1\mid X_i(t) = X_i \right\} =
 \left\{ 
  \begin{array}{ll}
p_i, & \text{If}~X_i=1 \\
1-p_i & \text{If}~X_i=0
  \end{array}\right.
  \end{array},
\end{align}
 and the second term can be computed as follows:
\begin{align}
 \begin{array}{ll}
 & \Pr\left\{ X_i(t) =X_i \mid  I(t)=I, o(t+1)=o, a(t)=a \right\} =
 \\& 
 \left\{ 
  \begin{array}{ll}\label{eq_blf_second}
       1,  & \text{if}~X_i=1;\delta_i = 1,~\bar{X}_i=1 \\
        1,  & \text{if}~X_i=0;\delta_i = 1,~\bar{X}_i=0 \\
b_i(t),  & \text{if}~X_i=1;\delta_i \neq 1 \\
1-b_i(t),  & \text{if}~X_i=0;\delta_i \neq 1 \\
0, & \text{otherwise}
         \end{array}\right.
    \end{array}.
\end{align}
Now expanding \eqref{eq_blf_finalsum} using the corresponding terms in  \eqref{eq_blf_fist} and in \eqref{eq_blf_second} yields the belief update formula in \eqref{eq_blf_evl}, which completes the proof.

\renewcommand{\baselinestretch}{0.9555}
\small
\bibliographystyle{ieeetr}
\bibliography{Bib_References/conf_short,
Bib_References/IEEEabrv,
Bib_References/Bibliography}

\end{document}